\documentclass[11pt,journal,onecolumn,letterpaper]{IEEEtran}
\usepackage{graphicx}
\usepackage{amssymb}

\usepackage{amsfonts}
\usepackage{cite}
\usepackage{subeqnarray}
\usepackage{amssymb}
\usepackage{amsmath}
\usepackage{graphics}
\usepackage{epsfig}
\usepackage{amsthm}
\usepackage{array}
\usepackage{arydshln}
\usepackage{subfigure}
\usepackage{multicol}
\usepackage{subfigure}
\usepackage{algorithmic}
\usepackage{algorithm}

 \newtheorem{thm}{Theorem}[subsection]
 \newtheorem{cor}[thm]{Corollary}
 \newtheorem{lemm}[thm]{Lemma}

\begin{document}

\title {A Cooperative MARC Scheme Using Analogue Network Coding to Achieve Second-Order Diversity}

\author{Mohammad Shahrokh~Esfahani,~\IEEEmembership{Student Member,~IEEE,}
                and~Masoumeh~Nasiri-kenari,~\IEEEmembership{Member,~IEEE}

\thanks{M. S. Esfahani is with the Department
of Electrical and Computer Engineering, Texas A\&M University, College Station,
TX, 77843.\protect\\
E-mail: m.shahrokh@tamu.edu}
\thanks{Masoumeh~Nasiri-kenari is with the Department
of Electrical and Computer Engineering, Sharif University of Technology, Tehran, Iran.\protect\\
E-mail: mnasiri@sharif.edu}}

\maketitle

\begin{abstract}
A multiple access relay channel (MARC) is considered in which an analogue-like network coding is implemented in the relay
node. This analogue coding is a simple addition of the received signals at the relay node. Using ``nulling detection'' structure employed in V-BLAST
receiver, we propose a detection scheme in the destination which is able to provide a diversity order of two for all users. We analytically evaluate the performance of our proposed scheme for the MARC with two users where tight upper bounds for both uncoded and Convolutionally coded transmission blocks are provided. We verify our
analytical evaluations by simulations and compare the results with those of non-cooperative transmission and Alamouti's scheme for the same power and rate transmission. Our results indicate that while our proposed scheme shows a comparable performance compared to the Alamouti's scheme, it substantially outperforms the non-cooperate transmission.
\end{abstract}

\begin{IEEEkeywords}
Cooperative communication, MARC, analogue network coding, V-BLAST detection, Alamouti space-time coding
\end{IEEEkeywords}
\maketitle

\section{Introduction}
It has been shown that a cooperative communication can provide more reliable communication by reducing the error probability, while the required power remains the same. The concept of cooperation in the communication has a long history going back to \cite{gaarder1975capacity} where an information theoretic aspect view of cooperative communication in a multiple access channel (MAC) is defined. Although this work was based on the concept of feedback, it can be seen as a beginning of the notion of cooperation in communication. This concept have then been further extended in the information theoretic framework in several seminal works in \cite{king1978multiple,cover1981achievable,aref1981information,willems1983discrete,willems1985discrete,ozarow1984capacity}.

The initial spark started in~\cite{aref1981information} then lead to a fundamental context of network coding (NC), as a method for achieving capacity of wired networks, introduced in 2000~\cite{ahlswede2000}. Thereafter, there have been a vast number of researches on this idea where different advantages of using NC have been addressed. Owing to the cooperative nature of network coding, in \cite{chen2006wireless}, it was shown that using NC over a wireless network can provide diversity, where it is called cooperative diversity in NC, similar to its previous counterparts \cite{laneman2001efficient,nosratinia2004cooperative,sendonaris2003user,sendonaris2003user2,stefanov2004cooperative,janani2004coded}. It then lead the researchers in wireless communication also to become interested in NC, either in ad-hoc or cellular networks \cite{laneman2004cooperative,xiao2007network,katti2008xors,yang2007network}. Moreover, implementing NC, as an alternative method of routing, in which the nodes in the network process their received packets, can help to increase the multiplexing gain of such networks.   

A very simple example of a cooperative system is a multiple access relay channel (MARC). Although there is no direct cooperation between the users, sharing the relay node can be seen as a cooperative behavior. A relay node can be utilized conventionally, where at each time slot the relay only serves for one of the users. On the other hand, if the relay can process the received data from the users and aggregate these data and then send the new aggregated data to the destination, it would be possible to achieve diversity.  One instance of such doing can be found in \cite{2}, where an XORing scheme was considered as a simple way of implementing NC in a wireless network in order to provide diversity in the network. The authors in ~\cite{3} have proposed a scheme in a MARC in which the users send the log-likelihood ratio (LLR) of their coded messages to the destination in a noisy channel, and they have shown that their proposed scheme improves the bit error probability (BEP) of the system. Practical analogue NC for cooperation in the physical-layer of MARC has not been considered yet, and it has prompted us to consider this method of
cooperative scheme here.

In this paper, we  propose a cooperative scheme based on a simple addition of the received signals in the relay node. We loosely call this scheme as an instance of analogue NC. To this end, we consider a system which uses Decode and Forward (DF) for relay cooperation, and after modulating the users' re-encoded packets, simply adds the packets for transmission in the relay. We assume that each node has a single antenna. Also we assume that each user has a power limitation of $P$ , but the power limitation of the relay node in  $n$-user MARC will be $nP$ which is logical. We propose a detection method in the end-node, i.e. destination, being based on V-BLAST model in a multiple-antenna system~\cite{4,5}. We derive a tight upper bound on the system's BEP for both uncoded and coded transmissions, from which we show that using V-BLAST detection technique, we can achieve full-order diversity, being two in a two-user MARC setting. Simulation results indicate that the bounds are tight enough. Moreover, results show that the proposed scheme can achieve a diversity order of two for any number of users and only one relay node , in contrast to the traditional methods (e.g. Alamouti scheme) in which each user must have two antennas for transmission.

The paper is organized as follows. In Section~\ref{sec:systemdescription}, we briefly describe the system and the proposed transmission protocol. In Section~\ref{sec:detection} we introduce the detection algorithm. The performance analysis is presented in Section~\ref{sec:perfanal}. Numerical experiments are shown in Section~\ref{sec:numexp}. We finally finish the paper with the concluding remarks in Section~\ref{sec:conc}.

Notation and abbreviations in the paper are as below. Boldface lower case letters denote column vectors. Concatenation of several vectors is also denoted by a boldface lower case letter. $k-th$ element of the vector $\mathbf{x}$ is denoted by $x_k$. Boldface upper case letters are used to denote matrices. $.^T$ denotes transpose operators. The cadinality of a set $S$ is shown by $|S|$. $\Pr(A)$ denotes the probability of event $A.$ When dealing with random variable, $x$, we use $F_X(x)$, $f_X(x)$, and $\phi_x(S)$ to denote the cumulative distribution function (CDF), probability density function (PDF), and characteristic function, respectively. Finally, the notation $\text{E}_{x}[g(x)]$ is used to denote taking expectation of $g(x)$ with respect to the subscript $x$. 


\section{System Description}\label{sec:systemdescription}


\subsection{System Model}

For the simplicity of presentation, we consider MARC with two users A and B, and one relay R, and the destination D. The results can be
generalized for more than two users sharing one relay. Figure~\ref{fig:marc-view} shows a two-user MARC used throughly in this paper. 
\def\scaleval{0.5}
\def\plotsizeap{5cm}

\begin{figure}[!hs]
\begin{center}
\includegraphics[scale=\scaleval]{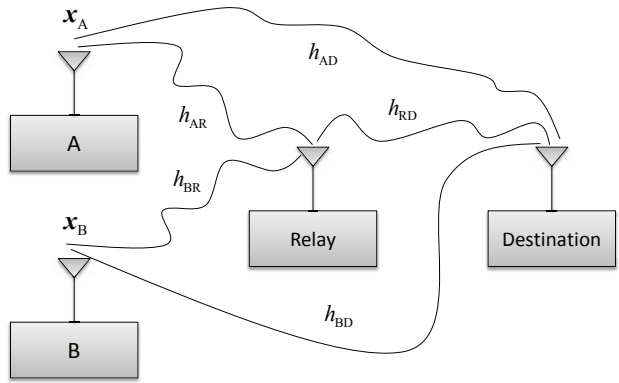}
\caption{An illustrative view of a multiple access relay channel with two users.}\label{fig:marc-view}
\end{center}
\end{figure}
In what follows, we model the received signals at the destination and relay node. The transmitted signal by the relay node is considered in the next subsection. 
Let ${{h_{t,r}}}$ denote the channel gain between nodes $\text{t}$ and $\text{r}$, where $\text{t}\in{{\{\text{A},\text{B},\text{R}}\}}$
and $\text{r}\in{\{\text{R,D}}\}$. The channel gains are assumed to be independent zero mean complex Gaussian random variables with the powers of
${{\Omega_{\rightarrow \text{R}}}}$ and ${{\Omega_{\rightarrow \text{D}}}}$, for the link between the user ($\text{A}$ or $\text{B}$) and the relay ($\text{R}$), and the link between node ($\text{A}$, $\text{B}$, or $\text{R}$) and the destination, respectively. That is, we have considered different qualities for user-relay and
user (relay)-destination channels. We assume a block fading in which the channel (gain) coefficient is constant during each block transmission \cite{tse2005fundamentals}. Moreover, we assume that each user sends a codeword $\mathbf{x}_{i} \in \mathcal{X}$. 

For all the wireless communication links shown in Figure~\ref{fig:marc-view}, we consider an additive white Gaussian noise (AWGN) with one-sided power spectral density (PSD) of ${{N_0}}$. Then, knowing the fact that $h_{t,r}$'s are normally distributed with the powers mentioned above, the instantaneous signal to noise ratio (SNR) of the link between user $\text{t}\in\{\text{A,B}\}$ and the relay node $\text{R}$ is given by $\text{SNR}_{tR}=
  \frac{1}{N_0}{P}{\left|{h_{t,R}}\right|^2}$ being distributed exponentially as follows

\begin{equation}\label{1}
 \text{SNR}_{tR}{\sim}{\exp(-\frac{1}{P}{N_0}{{\Omega_{\rightarrow R}}})}
\end{equation}
Furthermore, the received signals at the destination from two users can be modeled as
\begin{align}\label{eq:received dest}
&{\mathbf{y}_{\text{AD}}}={h_{\text{A,D}}}{\mathbf{x}_\text{A}}+{\mathbf{n}_\text{A}^d}\\
&{\mathbf{y}_{\text{BD}}}={h_{\text{B,D}}}.{\mathbf{x}_\text{B}}+{\mathbf{n}_\text{B}^d}
\end{align}
whereas ${\mathbf{n}_1}$ and ${\mathbf{n}_2}$  are the received additive (vector-valued) white Gaussian
noises at the destination side. The ${\mathbf{x}_{\text{B}}}$ and ${\mathbf{x}_{\text{A}}}$, the vectors of
dimension $n$, denote the transmitted blocks from users $\text{A}$ and $\text{B}$,
respectively, with the following power constraints:
\begin{align}\label{eq:powers}
&\frac{1}{n}{\sum\limits_{i=1}^n{{x_{\text{A},i}}^2}}=P\\
&\frac{1}{n}{\sum\limits_{i=1}^n{{x_{\text{B},i}}^2}}=P
\end{align}
where $n$ is the number of bits in each block. Similarly, the received signals at the relay node are given by\footnote{To keep generality, we assume different statistics for the additive noises corresponding to users' signals received at the relay and destination nodes.}
\begin{align}\label{eq:received dest}
&{\mathbf{y}_{\text{AR}}}={h_{\text{A,R}}}{\mathbf{x}_\text{A}}+{\mathbf{n}^{r}_\text{A}}\\
&{\mathbf{y}_{\text{BR}}}={h_{\text{B,R}}}.{\mathbf{x}_\text{B}}+{\mathbf{n}^{r}_\text{B}}
\end{align}
where $\mathbf{n}^{r}_\text{A}$ and $\mathbf{n}^{r}_\text{B}$ are the noises received at the relay node through user $\text{A}$ and $\text{B}$, respectively.

\subsection{Transmission Algorithm}

Since, in this paper we aim at proposing a cooperative protocol while sharing a single relay node by two nodes, we propose to divide the time as shown in Figure~\ref{fig:timeslot}. In Figure~\ref{fig:timeslot} the transmission time slots are allocated based on time division multiple access (TDMA) scheme. 
\begin{figure}[htp]
\begin{center}
\includegraphics[scale=\scaleval]{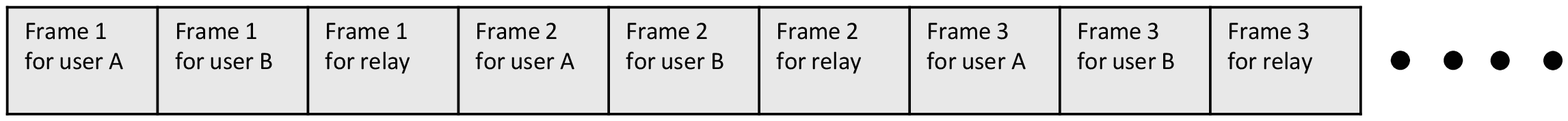}
\caption{The transmission time slots allocated for users and relay node.}\label{fig:timeslot}
\end{center}
\end{figure}
As illustrated in Figure~\ref{fig:timeslot}, in the first two time slots, users A and B transmit their blocks.
Then, the relay, based on its received signals ${\mathbf{y}_{\text{AR}}}$ and ${\mathbf{y}_{\text{BR}}}$, decodes the received blocks. Depending on
whether or not the relay can correctly decode the received blocks, which can be recognized by a error-detecting code (e.g., CRC codes), the relay decides upon cooperation with the two transmitters in the third slot, as shown in Figure~\ref{fig:timeslot} leading to four possible scenarios summarized in Table~\ref{table:relayscenarios}. 

\begin{table}[!h]
\begin{center}
\begin{tabular}{|c|c|c|}
\hline
Event& State Notation& Relay Action\\ \hline
Relay cannot decode none of its received blocks&${\text{s}_0}$ &Silence\\ \hline
Relay can decode only received block of user $A$& ${\text{s}_1}$& Transmitting $\hat{\mathbf{x}_{\text{A}}}$\\ \hline
Relay can decode only received block of user $B$&${\text{s}_2}$& Transmitting $\hat{\mathbf{x}_{\text{B}}}$\\ \hline
Relay can decode the received blocks of both users&${\text{s}_3}$ & Transmitting $\hat{\mathbf{x}_{\text{A}}}+\hat{\mathbf{x}_{\text{B}}}$\\ \hline
\end{tabular}
\end{center}
\caption{Four possible states at the relay node.}
\label{table:relayscenarios}
\end{table}

As can be inferred from Table~\ref{table:relayscenarios}, relay cooperation can be categorized into four states:
First, corresponding to ${\text{s}_0}$ , the relay
will be silent, and as a result the destination only receives the
direct block transmissions from two users. Hence, the bit error probability
is the same as that of a direct transmission with power ${P}$.
Second, corresponding to ${\text{s}_1}$  and ${\text{s}_2}$, the relay node only cooperates in transmission of one of the
users. Thereby, only one of the users benefits a second order diversity
transmission. Denoting this user by $u$ ($u\in{{\{\text{A,B}}\}}$),
the received signal at the destination from the relay will be
\begin{equation}\label{eq:case2-received} 
{\mathbf{y}_{\text{RD}}}={h_{\text{R,D}}}{\hat{\mathbf{x}}_u}+{\mathbf{n}_{\text{R}}^d},
\end{equation} 
where $\hat{\mathbf{x}}_u$ and $\mathbf{n}^r_{\text{D}}$ are relay's estimate of user's signal and additive noise at the destination, respectively. In this case, the bit error probability for the assisted user is equal to that of a Maximum Ratio Combining (MRC) receiver, \cite{chen2005analysis}, with two
independent receptions of the transmitted signal, and for the other user, it is similar to the first state.
Lastly, for ${\text{s}_3}$ , the relay can successfully decode both users' blocks. Then, the relay encodes and
modulates these signals separately, and transmits the sum of the
modulated analog signals. Thus, the destination node will receive the following signal
\begin{equation}\label{eq:received-dest-from-relay} {\mathbf{y}_{\text{RD}}}={h_{\text{R,D}}}{(\hat{\mathbf{x}}_{\text{A}}+\hat{\mathbf{x}}_{\text{B}})}+{\mathbf{n}_{\text{R}}^d}.
\end{equation}






\section{V-BLAST-Based Detection Algorithm}\label{sec:detection}

In this section, we propose a method for detection at the destination node. First, we assume that the destination knows fading coefficients, e.g. by acquiring
transmitting pilot symbols, and also knows the state of cooperation at the relay, e.g. it can be implemented by sending 2 bits encoding four states. For $\text{s}_0$ in Table~\ref{table:relayscenarios}, both users have interference-free direct transmission, and, as stated above, the conventional receiver is used. For $\text{s}_1$ ($\text{s}_2$), the user which is able to utilize relay resource, the destination uses MRC of the two received signals, one from the user and the other from the relay. In the last state,
we propose a novel method for detection relying on the basis of V-BLAST detection. First, from equations~\eqref{eq:received dest}-\eqref{eq:received-dest-from-relay}, we rewrite the three received signals in three time slots (see Figure~\ref{fig:timeslot}) at the destination in a following matrix format given by

\begin{equation}\label{eq:whole-received-matrix}
 \left( {\begin{array}{*{20}{c}}
   {\mathbf{y}_{\text{AD}}}  \\
 {\mathbf{y}_{\text{BD}}}  \\
  {\mathbf{y}_{\text{RD}}} \\
\end{array}} \right)= {\left( {\begin{array}{*{20}{c}}
   {h_{\text{A,D}}}  \\
   {0}  \\
   {h_{\text{R,D}}}  \\
\end{array}} \right)}{\mathbf{x}_{\text{A}}}
+{\left( {\begin{array}{*{20}{c}}
   {0}  \\
   {h_{\text{B,D}}}  \\
   {h_{\text{R,D}}}  \\
\end{array}} \right)}{\mathbf{x}_{\text{B}}}+{\left( {\begin{array}{*{20}{c}}
   {\mathbf{n}_{\text{A}}^d}  \\
   {\mathbf{n}_{\text{B}}^d}  \\
   {\mathbf{n}_{\text{R}}^d}  \\
\end{array}} \right)}\end{equation}

Without loss of generality, we focus on detecting user $\text{A}$'s block. Doing so, we first define the vector conveying user $\text{B}$'s as follows
\begin{equation}\label{10}
{\mathbf{u}_{\text{B}}}= \frac{1}{\sqrt{{{h_{\text{B,D}}}^2}+{{h_{\text{R,D}}}^2}}}{\left(
{\begin{array}{*{20}{c}}
   {0}  \\
   {h_{\text{B,D}}}  \\
   {h_{\text{R,D}}}  \\
\end{array}} \right)}.
\end{equation}
The null space of ${\mathbf{u}_{\text{B}}}$, denoted by $\mathbf{U} \buildrel \Delta \over = \left[ {\mathbf{u}_{\text{1}} \ {\rm{  }}\ \mathbf{u}_{\text{2}} } \right]$, can then be spanned by
\begin{equation}\label{eq:nullspace}
\begin{array}{l}
\mathbf{u}_{\text{1}}  \buildrel \Delta \over = \frac{1}{{\sqrt {|h_{B,D} |^2  + |h_{R,D} |^2 } }}\left( \begin{array}{c}
 0 \\
  - h_{\text{R,D}}^*  \\
 h_{\text{B,D}}^*  \\
 \end{array} \right){\rm{                                    }} ,\\
\mathbf{u}_{\text{2}}  \buildrel \Delta \over = \left( \begin{array}{c}
 1 \\
 0 \\
 0 \\
 \end{array} \right) .\\
 \end{array}
\end{equation}

It should be noted that, the matrix $\mathbf{U}$ defined in equation~\eqref{eq:nullspace} is one of many possible
null spaces for vector  $\mathbf{u}_{\text{A}}$. Now multiplying both sides of~\eqref{eq:whole-received-matrix} by transposed of matrix $\mathbf{U}$, one may obtain
\begin{equation}\label{eq:receive-signal-a}
\tilde{\mathbf{y}} \buildrel \Delta \over = \mathbf{U}^T \left( \begin{array}{c}
   {\mathbf{y}_{\text{AD}}}  \\
 {\mathbf{y}_{\text{BD}}}  \\
  {\mathbf{y}_{\text{RD}}} \\ \end{array} \right) = \left[ \begin{array}{c}
 \frac{{h_{\text{R,D}} h_{\text{B,D}} }}{{\sqrt {|h_{\text{B,D}} |^2  + |h_{\text{R,D}} |^2 } }} \\
 h_{\text{A,D}}  \\
 \end{array} \right]\mathbf{x}_{\text{A}}  + \left[ \begin{array}{c}
 \frac{{ - h_{\text{R,D}} }}{{\sqrt {|h_{\text{B,D}} |^2  + |h_{\text{R,D}} |^2 } }}\mathbf{n}_{\text{B}}^d  + \frac{{h_{\text{B,D}} }}{{\sqrt {|h_{\text{B,D}} |^2  + |h_{\text{R,D}} |^2 } }}\mathbf{n}_{\text{R}}^d  \\
 \mathbf{n}_{\text{A}}^d  \\
 \end{array} \right]{\rm{   }}
\end{equation}

As can be realized from equation~\eqref{eq:receive-signal-a}, we have two receptions for user $\text{A}$'s
signal. Owing to independency between $\mathbf{n}_{\text{A}}^d$, $\mathbf{n}_{\text{B}}^d$, and $\mathbf{n}_{\text{R}}^d$, we have

\begin{equation}
\begin{array}{l}
 \mathbf{n}^{\prime}_1  \buildrel \Delta \over = \frac{{ - h_{\text{R,D}} }}{{\sqrt {|h_{\text{B,D}} |^2  + |h_{\text{R,D}} |^2 } }}\mathbf{n}_{\text{B}}^d  + \frac{{h_{\text{B,D}} }}{{\sqrt {|h_{\text{B,D}} |^2  + |h_{\text{R,D}} |^2 } }}\mathbf{n}_{\text{R}}^d  \sim \mathcal{CN}(\mathbf{o},N_0\mathbf{I}_{n\times n}) \\
 \mathbf{n}^{\prime}_2  \buildrel \Delta \over = \mathbf{n}_{\text{A}}^d  \sim \mathcal{CN}(\mathbf{0},N_0\mathbf{I}_{n\times n} ) \\
 \end{array}
\end{equation}
Now, if we use the MRC method for detection of ${\mathbf{x}_{\text{A}}}$, we obtain

\begin{equation}
\tilde{\mathbf{y}}_{\text{A}}  \buildrel \Delta \over = (|h_{\text{A,D}} |^2  +
\frac{{|h_{\text{B,D}} h_{\text{R,D}} |^2 }}{{|h_{\text{B,D}} |^2  + |h_{\text{R,D}} |^2 }})\mathbf{x}_{\text{A}}
+ (\frac{{(h_{\text{R,D}} h_{\text{B,D}})^* }}{{|h_{\text{B,D}} |^2  + |h_{\text{R,D}} |^2
}}\mathbf{n}^{\prime}_1 + h_{\text{A,D}}^* \mathbf{n}_2^{\prime}  )
\end{equation}
A similar procedure can be employed for ${\mathbf{x}_{\text{B}}}$ by
multiplying two sides of equation~\eqref{eq:whole-received-matrix} by another matrix representing
the null space of the vector, conveying user $\text{A}$'s signal (similar to that of user $\text{B}$).
Then, the SNRs at the outputs of MRC receivers can be written as
\begin{equation}
\begin{array}{l}
 \text{SNR}_{\text{AD}} = (|h_{\text{A,D}} |^2  + \frac{{|h_{\text{B,D}} h_{\text{R,D}} |^2 }}{{|h_{\text{B,D}} |^2  + |h_{\text{R,D}} |^2 }})\frac{1}{N_0}P {\rm{  }} \\
 \text{SNR}_{\text{BD}}  = (|h_{\text{B,D}} |^2  + \frac{{|h_{\text{A,D}} h_{\text{R,D}} |^2 }}{{|h_{\text{A,D}} |^2  + |h_{\text{R,D}} |^2 }})\frac{1}{N_0}P {\rm{     }} \\
 \end{array}
\end{equation}
Computing these factors, the destination node chooses to decode the
user with the maximum SNR at first. Let ${u_{\text{max}}}$ denote the user with
the highest SNR, i.e.
\begin{equation}\label{eq:highest-snr}
\text{SNR}_{u_{\text{max}}^1\text{D}}  = \max \{ (|h_{\text{B,D}} |^2  + \frac{{|h_{\text{A,D}} h_{\text{R,D}} |^2 }}{{|h_{\text{A,D}} |^2
+ |h_{\text{R,D}} |^2 }})\frac{P}{N_0}, (|h_{\text{A,D}}|^2  + \frac{{|h_{\text{B,D}} h_{\text{R,D}} |^2
}}{{|h_{\text{B,D}} |^2  + |h_{\text{R,D}} |^2 }})\frac{P}{N_0}{\rm{ \}  }}
\end{equation}
For this user, the destination utilizes the MRC receiver. Then, the
effect of this user signal is subtracted from the received signal
${Y_{\text{RD}}}$, i.e. 
\[
{Y_{\text{RD}}}-h_{u_{\text{max}},\text{D}}\mathbf{x}_{u_{\text{max}}}.
\]
Then, conditioned on the detection of
the user $u_{\text{max}}$ being error-free, for the second user (denoted by
${u_{max}^{2}}$) a diversity order of two can be achieved, with the
following SNR
\begin{equation}\label{eq:second-highest-snr}
\text{SNR}_{{u_{max}^{2}}\text{D}}  = (|h_{_{{{{u_{max}^{2}}}},\text{D}} } |^2  + |h_{\text{R,D}} |^2 )\frac{P}{N_0}
\end{equation}
\begin{cor}[Generalization] Denoting the set of all users sharing a relay node in a MARC by $\mathcal{U}$, the SNR of the user to be decoded (at the destination) first is given by
\begin{equation}\label{eq:cor-general}
\text{SNR}_{u_{\text{max}}^1\text{D}}  = \underset{{u \in {\mathcal{U}}}}{\max}\text{ }
\Big[|h_{u,\text{D}} |^2  + \frac{1}{\underset{{\scriptstyle v \in {{\mathcal{U}}\backslash u}}}{\sum}{\frac{1}{{|h_{v,\text{D}} |^2 }} + \frac{1}{{|h_{\text{R,D}} |^2 }}} }\Big]\frac{P}{N_0}
\end{equation}
\end{cor}
After decoding the first user's signal, its effect will be subtracted from $\mathbf{y}_{\text{RD}}$. Then, the next best user is found, similar to equation~\eqref{eq:cor-general} except the fact that, now we have one user less the initial case. A detailed description of the proposed detection algorithm is illustrated in Algorithm~\ref{alg:detection}.  
\begin{algorithm}[!h]\small
\caption{V-BLAST approach for detection of multi-user MARC}
\begin{algorithmic}\label{alg:detection}
\STATE {{\textbf{Input:}} Channel Coefficients ($h_{u,\text{D}},u\in\mathcal{U}$ and $h_{\text{R,D}}$), $\mathbf{y}_{u\text{D}},u\in\mathcal{U}$, $\mathbf{y}_{\text{RD}}$}
\STATE{{\textbf{Output:}} $\mathbf{x}_u, u\in\mathcal{U}$}
\STATE{{\textbf{Initialize:}} $\mathcal{U}_{new}\leftarrow \mathcal{U}$, $\mathbf{y}^{new}_{\text{RD}}\leftarrow\mathbf{y}_{\text{RD}}$}
\FOR{$i=1$ to $|\mathcal{U}|$}
\STATE{Find the user with ${u_{\text{max}}^{i}}  = \underset{{u \in {\mathcal{U}}_{new}}}{\arg\max}\text{ }
\Big[|h_{u,\text{D}} |^2  + \frac{1}{\underset{{\scriptstyle v \in {{\mathcal{U}}\backslash u}}}{\sum}{\frac{1}{{|h_{v,\text{D}} |^2 }} + \frac{1}{{|h_{\text{R,D}} |^2 }}} }\Big]$}
\STATE{Decode the signal of $u_{\text{max}}^i$, $\hat{\mathbf{x}}_{{u_{max}^{i}}}$ using nulling technique for two signals $\mathbf{y}_{\text{RD}}^{new}$ and $\mathbf{y}_{{u_{max}^{i}}\text{D}}$}
\STATE{$\mathbf{y}_{\text{RD}}^{new}\leftarrow \mathbf{y}_{\text{RD}}-\hat{\mathbf{x}}_{u_{max}^i}h_{{u_{max}^{i}}}$}
\STATE{$\mathcal{U}_{new} \leftarrow \mathcal{U}_{new}\backslash {u_{max}^{i}}$}
\ENDFOR
\STATE{Return $\mathbf{x}_u,u\in\mathcal{U}$}
\end{algorithmic}
\end{algorithm}

\section{Performance Analysis in the Error-Propagation-Free state}\label{sec:perfanal}

In this section, we evaluate the performance of the proposed
receiver. We assume that the user-relay channel (URC) is stronger
than both the user- and relay-destination channels (that is ${{\Omega_{\rightarrow \text{R}}}}  >
{{\Omega_{\rightarrow \text{D}}}}$ ). We also assume that whenever the relay cooperates, it
cooperates correctly. This assumption sounds, since if a strong CRC
(Cyclic Redundancy Check) code is used, the relay can can detect all users' signals. Hence, if the relay detects an error
in its received codeword, it discards the whole block and does not
cooperate. In the following, we assume a BPSK
modulation.\footnote{The proposed framework can work for any modulation.
For simplicity of performance evaluation, we have considered a BPSK
modulation.} For the BEP evaluation, we first compute the
probabilities of each state summarized in Table~\ref{table:relayscenarios} as follows
\begin{equation}\label{eq:state-probabilities-defn}
\begin{array}{l}
 p_0  = \Pr (s_0) = p_{\text{AR}}p_{\text{BR}}  \\
 p_1  = \Pr (s_1) = (1 - p_{\text{AR}})p_{\text{BR}}  \\
 p_2  = \Pr (s_2) = p_{\text{AR}}(1 - p_{\text{BR}}) \\
 p_3  = \Pr (s_3) = (1 - p_{\text{AR}})(1 - p_{\text{BR}}), \\
 \end{array}\end{equation}
where  ${p_{AR}}$ and ${(p_{BR})}$ are the probabilities that the relay
cannot correctly decode user A  (B) transmitted block. These
probabilities can be assumed to be negligible at high SNRs. In this
section, we derive the BEP for an uncoded transmission, and then we
generalize our result to a convolutionally coded transmission, in
which the user's data are first convolutionally coded and then the
coded block is transmitted. 

Here we derive the BEP for user $u\in\mathcal{U}$ in the illustration in Figure~\ref{fig:marc-view}. Due to the symmetry in the structure, it will be the same for user B. We state the main result for the non-coded scheme in the following theorem:
\begin{thm}\label{thm:result1} Let us denote the BEP for user $u\in\mathcal{U}$ in a MARC by $P_{b,u}$. Moreover, assume that the relay cooperation is error-free. Then, for i.i.d. Rayleigh block fading channels between users (and
relay) to destination with parameter $ {\Omega_{\rightarrow \text{D}}},$ we have
\begin{equation}\label{eq:thm1}
\begin{array}{ll}
P_{b,u}&\leq \frac{1}{2}\Big[(p_0  + p_2 )\frac{\Omega_{\rightarrow \text{D}} }{{\text{snr}_0 + \Omega_{\rightarrow \text{D}} }} + p_1 \left( {\frac{\Omega_{\rightarrow \text{D}} }{{\text{snr}_0 + \Omega_{\rightarrow \text{D}}}}} \right)^2 \Big] + \frac{p_3}{4}\Big[(\frac{\Omega_{\rightarrow \text{D}} }{{\Omega_{\rightarrow \text{D}}  + \text{snr}_0}}\sqrt \pi  [\frac{{24\Omega_{\rightarrow \text{D}}}}{{{{\Gamma (2.5)}}(4\Omega_{\rightarrow \text{D}}  + \text{snr}_0)^3 }}F\left( {3,1.5;2.5;\frac{\text{snr}_0}{{4\Omega_{\rightarrow \text{D}}  + \text{snr}_0}}} \right) \\
  &+ \frac{4}{{{{\Gamma (2.5)}}(4\Omega_{\rightarrow \text{D}} + \text{snr}_0)^2 }}F\left( {2,0.5;2.5;\frac{\text{snr}_0}{{4\Omega_{\rightarrow \text{D}}  + \text{snr}_0}}} \right)] + (\frac{\Omega_{\rightarrow \text{D}} }{{\Omega_{\rightarrow \text{D}}  + \text{snr}_0}})^2 {\rm{)}}\Big]\\
 \end{array}
\end{equation}
where $\text{snr}_0:=\frac{P}{{N_0 }}$ and $ F(a,b;c;z)$ is the generalized Hypergeometric function \cite{integraltables}
and $p_i$'s are defined in equation~\eqref{eq:state-probabilities-defn}.
\end{thm}
\begin{proof} 
Please refer to Appendix~\ref{app:thmresult1}. 
\end{proof} 
In the following theorem, we generalize our results for a system which uses the same convolution codes for source and relay transmissions.

\begin{thm}\label{thm:result2} For a Convolutionally coded system using a code with rate $r$  and free distance ${d_{free}}$ and $k$ bits in each coded transmitted block, when utilizing a hard-input Viterbi decoder, at high SNR, the bit error rate(BER) has the following  bound
\begin{equation}\label{eq:thm2}
\begin{array}{ll}
 P_{b,u} &\leq \alpha_{c}\{(\frac{1}{2}\Big[(p_0  + p_2)\frac{\Omega_{\rightarrow D} }{{\text{snr}_0^c + \Omega_{\rightarrow D}}} + p_1\left( {\frac{\Omega_{\rightarrow D} }{{\text{snr}_0^c + \Omega_{\rightarrow D} }}} \right)^2 \Big] {\rm{ }} + \frac{p_3}{4}\Big[(\frac{\Omega_{\rightarrow D} }{{\Omega_{\rightarrow D}  + \text{snr}_0^c}}\sqrt \pi  [\frac{{4\Omega_{\rightarrow D} }}{{\Gamma (2.5)(24\Omega_{\rightarrow D}  + \text{snr}_0^c)^3 }}F( {3,1.5;2.5;\frac{\text{snr}_0^c}{{4\Omega_{\rightarrow D}  + \text{snr}_0^c}}})\\& + \frac{1}{{(4\Omega_{\rightarrow D}  + \text{snr}_0^c)^2 }}\frac{4}{{\Gamma (2.5)}}F({2,0.5;2.5;\frac{\text{snr}_0^c}{{4\Omega_{\rightarrow D}  + \text{snr}_0^c}}})] + (\frac{\Omega_{\rightarrow D} }{{\Omega_{\rightarrow D}  + \text{snr}_0^c}})^2 \Big]\}\\
 \end{array}
\end{equation}
where $\alpha_{c}=\frac{{2^{\frac{{d_{free} }}{2}} .B_{d_{free} } }}{k}$ and $\text{snr}_0^c:=\frac{P}{{N_0 }} \times \frac{{rd_{free} }}{2}$, in which ${B_{d_{free}}}$ is the number of codewords with weight
${d_{free}}$.
\end{thm}
\begin{proof} Please refer to Appendix~\ref{app:thmresult2}.\end{proof}

\begin{cor}[Asymptotic behvior]
For large SNRs (large values of $\text{snr}_0$ or $\text{snr}_0^c$), as  ${ \text{snr}_0 (\text{snr}_0)^c \to \infty} $, we have
\begin{equation}
\mathop {\lim }\limits_{\text{snr}_0 \to \infty } P_{b,u}  \leq \mathop {\lim
}\limits_{\text{snr}_0\to \infty } \{ \alpha (p_0  +
p_2)\frac{\lambda }{{\text{snr}_0 + \lambda }} +
(\beta p_1 + \gamma p_3)\left( {\frac{1}{{\text{snr}_0 + \lambda }}} \right)^2  +
p_3 o\left( {\left( {\frac{1}{{\text{snr}_0 + \lambda
}}} \right)^2 } \right)\} {\rm{ }},
\end{equation}
where $\alpha$, $\beta$, and $\gamma$ are some positive constants, and $g(x)=o({f(x)})$ means $\mathop
{\lim }\limits_{x \to \infty } \frac{f(x)}{g(x)}=0$.
\end{cor}
Thus, since at large SNRs, ${\rm p_0}$ and ${\rm p_1}$ will be negligible compared to ${\rm p_3,}$ the dominant terms will be the second and third terms which decrease inversely with squared SNR. As a result, with the definition of the diversity order, the proposed scheme can achieve second order of diversity.

\section{Numerical Experiments}\label{sec:numexp}

In this section, we present some numerical results to evaluate the performance of the proposed scheme. We set $\Omega_{\rightarrow \text{R}}=1$ and
change $\Omega_{\rightarrow \text{D}}$ to consider various link qualities (it will be discussed later). First we investigate the uncoded scheme. Figure~\ref{fig:simul1} shows the plots of BEP versus SNR per bit for direct transmission, Alamouti scheme using 2 transmitting antennas, and our proposed scheme. For our scheme, the plot based on the analytical result of Theorem 1 has been included as well. In this plots, we have assumed ideal user and relay channel (URC). That is users to relay links are error free.
\begin{figure}[htp]
\centering
\includegraphics[width=0.8\textwidth]{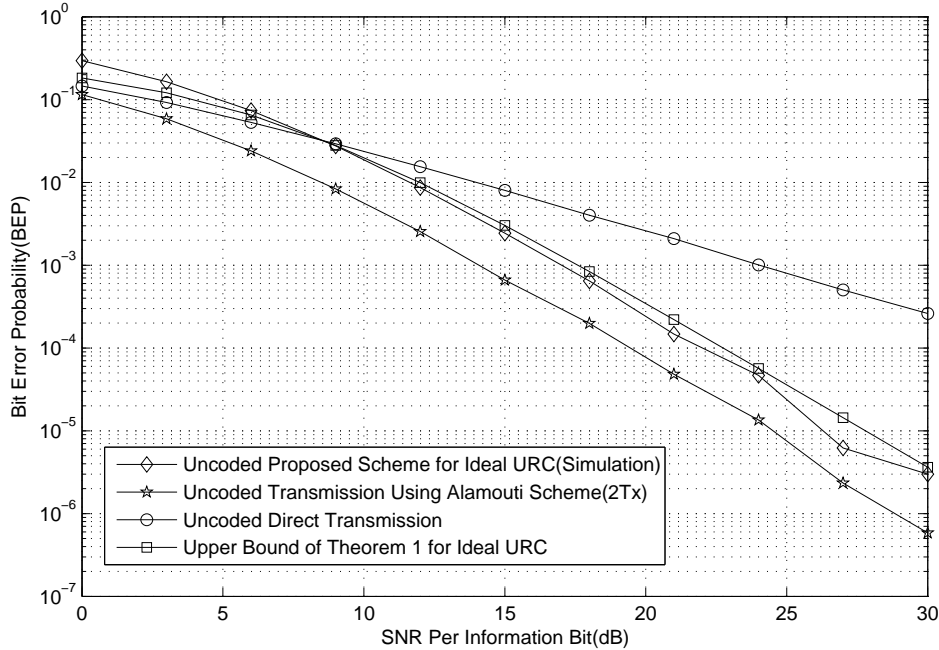}
\caption{The plots of BEP versus SNR per bit for uncoded proposed MARC with ideal URC, uncoded direct transmission, and uncoded Alamouti transmission.}\label{fig:simul1}
\end{figure}
As we expect, for the proposed scheme, the simulation result coincides with the bound of Theorem 1 at high SNR, and the slope of the plot is the same as that of the Alamouti scheme which indicates that the proposed
scheme achieves second order of diversity.

 For the rest of the comparisons, we consider coded schemes in which we
assume that users and relay use the same convolutional code for transmission where each block consists of
50 information bits, i.e., $K=50$ . We assume that the users A and B
apply CRC to their transmitted block in order the relay to be able
to detect its decoders' errors, also a convolutional code with
$rate=1/3$ and $ {d_{free}  = 8}$ which its octal coefficients are
${[{5, \rm{ 7, 7]}} } $ is considered for the coded transmission.

First, we show that the upper bound derived in Theorem 2 is tight
enough. Figure~\ref{fig:simul2} shows the performance of our proposed scheme based on
both simulation and analytical results in the presence of ideal URC.
\begin{figure}[htp]
\centering
\includegraphics[width=0.8\textwidth]{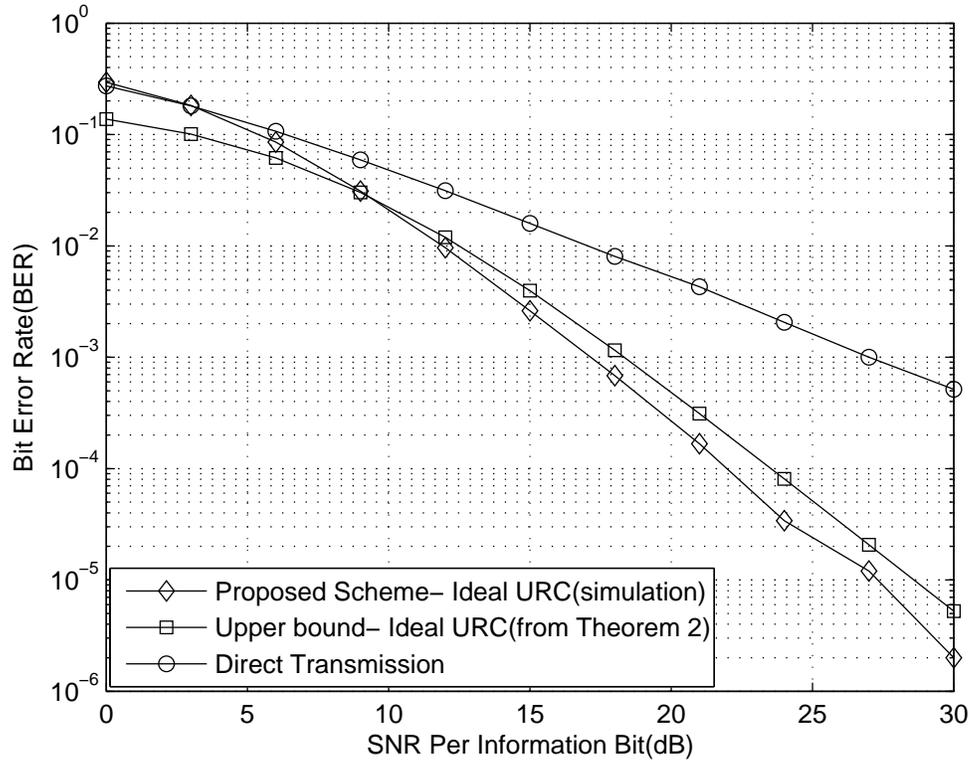}
\caption{The plots of BER versus SNR per bit for MARC with ideal URC
and direct transmission.}\label{fig:simul2}
\end{figure}

\begin{figure}[htp]
\centering
\includegraphics[width=0.8\textwidth]{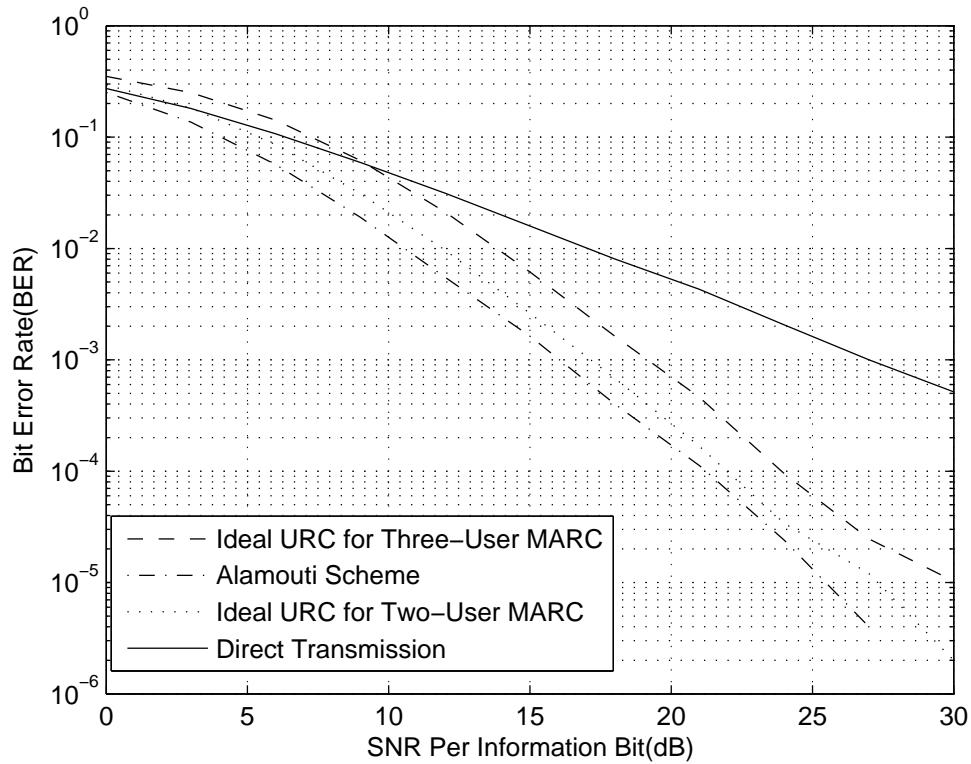}
\caption{Plots of BER versus SNR for three-user MARC, two-user MARC
with ideal URC, direct transmission, and Alamouti
scheme.}\label{fig:erptsqfit}
\end{figure}
For comparison, the performance of direct transmission (without the relay cooperation) and the Alamouti scheme(for two transmit antennas) are provided. We consider equal total transmission powers for comparison. It means that the total power of the user and relay transmissions is the same as the power
used in direct and Alamouti transmissions. As can be realized, the derived upper bound is tight enough. Furthermore, the proposed scheme substantially outperforms the direct transmission. Figure~\ref{fig:simul3} shows the plots of BER versus SNR per bit for our proposed MARC scheme with different URC's qualities based on both simulations and the upper bound derived in Theorem 2. 
\begin{figure}[htp]
\centering
\includegraphics[width=0.8\textwidth]{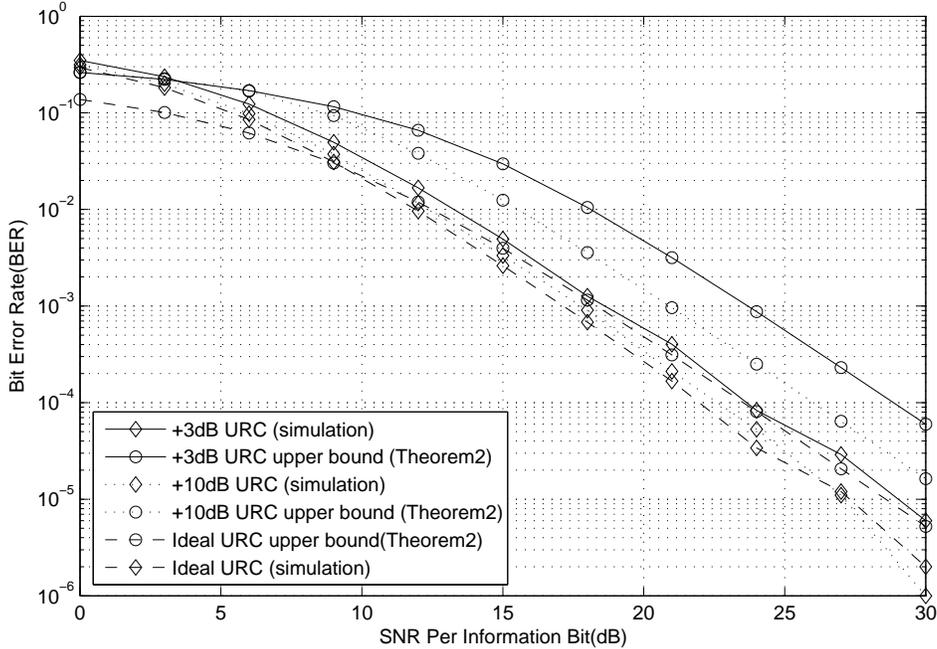}
\caption{The plots of BER versus SNR for different qualities of
URC.}\label{fig:simul3}
\end{figure}
In this plot, for instance the notion "$+3dB$ URC", indicates that the qualities of the sources A and B to the relay links are 3dB better than the sources to the
destination links. From this figure, the upper bound becomes looser when the quality of URC decreases. However, it must be noted that the slope of the plot does not change with the URC's qualities. 

In Figure~\ref{fig:simul4} 6, we compare the performance of our proposed scheme for different qualities of URC with Alamouti scheme and direct transmission (without relay) for the same transmitted powers, based on simulation results. As can be generalized, the slope of BER plots for our scheme for different URC's qualities are the same as the Alamouti scheme which indicates a diversity of order two, as expected from our theoretical derivations.
\begin{figure}[htp]
\centering
\includegraphics[width=0.8\textwidth]{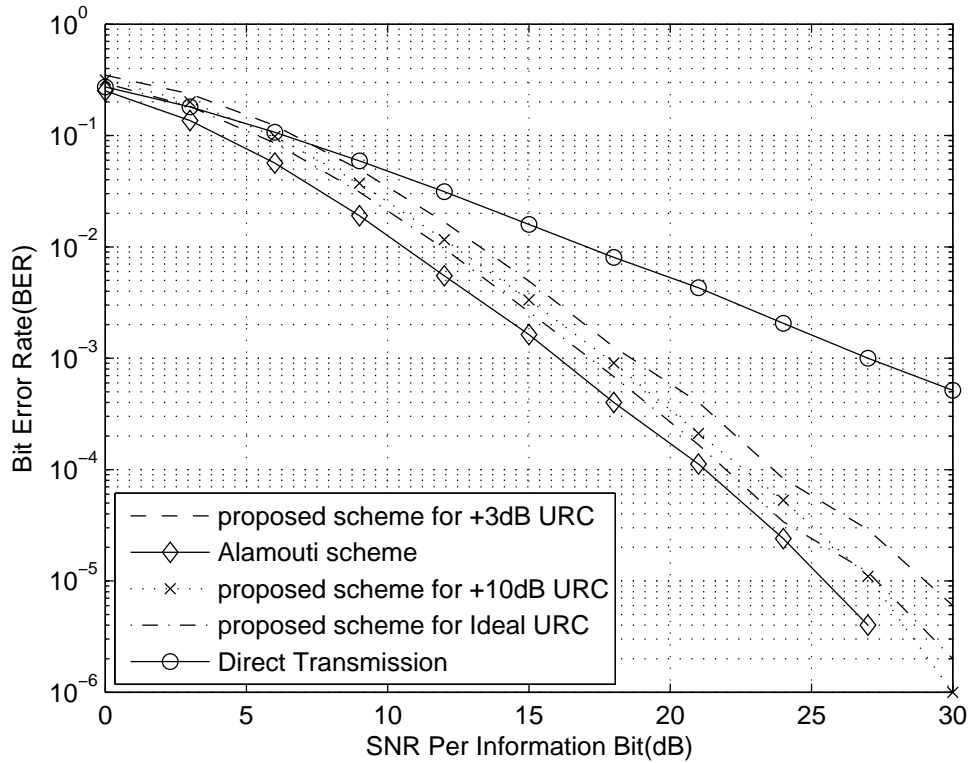}
\caption{The plots of BER versus SNR for different qualities of
URC, direct transmission, and Alamouti scheme.}\label{fig:simul4}
\end{figure}
In our further evaluation, we have also considered a MARC system with three users. We have assumed a ideal URC. Fig. 7 shows the plots of BER versus SNR per bit for MARC with 2 and 3 users. The plots for the Alamouti's scheme and no relay cooperation scheme are included for comparisons. As can be realized again, with only one
relay, the MARC with 3 users achieves second order of diversity. It must be noted that our scheme for 3 users, by 4 antennas can achieve second order of diversity for all users in contrast of the traditional schemes in which 6 antennas are required for the same number of users to achieve the diversity order of 2.


\section{Conclusion}\label{sec:conc}

We have introduced a novel scheme for MARC, which utilizes the
network coding in an analogue space. Using the nulling detection
structure applied in V-BLAST, destination can extract the users'
signals from two independent received paths. We have evaluated the
performance of the proposed receiver and provided a tight upper
bound for its bit error probability. Our analysis for i.i.d. block
Rayleigh fading channels have indicated that the proposed scheme can
achieve a diversity order of two. We have also provided a tight
upper bound of bit error rate for the coded transmission in which
the users and relay use the same Convolutional code. As have shown
for $n=3$ our proposed scheme can be generalized for a MARC with more
than two users and only one relay, in which all users benefit a
diversity order of two, in contrast to the traditional non-cooperated
schemes in which $2n$ antennas are required for the same order of
diversity, where $n$ is the number of users.
\appendices
\section{Proof of Theorem~\ref{thm:result1}}\label{app:thmresult1}

In what follows, without loss of generality, we find the BEP for user A. Then, considering the $j-$th bit of user A's codeword, the BEP may be expressed using the conditional bit error probabilities as follows
\begin{equation}
P_{b,\text{A}}  = \text{E}\Big[\Pr (\hat x_{\text{A},j}  \ne x_{A,j})\Big] = \text{E}\Big[\sum\limits_{i =
0}^3 {\Pr (\hat x_{\text{A},j}  \ne x_{\text{A},j} |s_i )} \Pr (s_i )\Big],
\end{equation}
in which $s_i$ stands for the rely's state, are defined in Tabel~\ref{table:relayscenarios} and the expectation is taken on channel gains. As we mentioned in previous
section, bit error probabilities of the system conditioned on ${s_0}$ and ${s_2}$ are the same for user A, therefore we have
\begin{equation}\label{eq:s0-prob}
\Pr (\hat x_{\text{A},j}  \ne x_{\text{A},j} |S = s_0 ) = \Pr (\hat x_{\text{A},j}  \ne
x_{\text{A},j} | s_2 ) = \text{E}\Big[Q(\sqrt {2|h_{\text{A,D}}|^2 P/N_0 })\Big].
\end{equation} 
Then, for ${s_1}$, as there exists also relay transmission, we have

\begin{equation}\label{eq:s1-prob}
\Pr (\hat x_{\text{A},j}  \ne x_{\text{A},j} |s_1 ) = \text{E}\Big[Q(\sqrt {2(|h_{\text{A,D}} |^2  +
|h_{\text{R,D}}|^2)P/N_0 } )\Big]
\end{equation}

For state  ${s_3}$ , depending on whether or not the user with the
highest received SNR at the MRC receiver's output (${{u_{max}^{1}}}$
defined in equation~\eqref{eq:highest-snr}, or ${{u_{max}^{2}}}$ defined in equation~\eqref{eq:second-highest-snr}), is user A, we will have different values for BEP.
Because of the symmetrical structure of the scheme we have
\begin{equation}\label{eq:prior-user-prob}
\Pr ({u_{max}^{1}}= A) = \Pr ({u_{max}^{2}}  = A) = \frac{1}{2}
\end{equation}
As a result, we simply obtain
\begin{equation}\label{eq:s3-prob}
\Pr (\hat x_A  \ne x_A |S = s_3 ) = \text{E}(\frac{1}{2}Q(\sqrt {2\text{SNR}_{u_{max}^{1}
}}) + \frac{1}{2}Q(\sqrt {2\text{SNR}_{{u_{max}^{2}} } } )) ,\end{equation}
where from equations~\eqref{eq:highest-snr} and \eqref{eq:second-highest-snr}, we have
\begin{equation}
\begin{array}{l}
\text{SNR}_{u_{max}^1}  = (|h_{\text{A,D}} |^2  + \frac{{|h_{\text{B,D}} h_{\text{R,D}} |^2 }}{{|h_{\text{B,D}} |^2  + |h_{\text{R,D}} |^2 }})P/N_0  \\
\text{SNR}_{{u_{max}^{2}}}  = (|h_{\text{A,D}} |^2  + |h_{\text{R,D}} |^2 )P/N_0  \\
 \end{array}
\end{equation}
Now, we compute equations . Plugging ${Q(x) \le \frac{1}{2}e^{ - \frac{{x^2 }}{2}}}$ in equations \eqref{eq:s0-prob}, \eqref{eq:s1-prob}, and \eqref{eq:s3-prob}, and taking the expectation, we obtain
\allowdisplaybreaks \begin{subeqnarray}\slabel{eq:subeq-prob-s02}
\Pr (\hat x_{\text{A},j}  &\ne& x_{\text{A},j} |s_0 ) = \Pr (\hat x_{\text{A},j}  \ne
x_{\text{A},j} | s_2 ) \leq \frac{1}{2}\phi _{|h_{\text{A,D}} |^2 }(P/N_0 )\\ \slabel{eq:subeq-prob-s1}
\Pr (\hat x_{\text{A},j}  &\ne& x_{\text{A},j} | s_1 ) \leq \frac{1}{2}\phi _x
(P/N_0)\\ \slabel{eq:subeq-prob-s3}
\Pr (\hat x_{\text{A},j}  &\ne& x_{\text{A},j} |S = s_3 ) \leq \frac{1}{4}[\phi_x
(P/N_{0}) + \phi _y (P/N_{0})]
\end{subeqnarray}
where $ { \phi _{|h_{\text{A,D}} |^2 } }$, ${ \phi _X (.) } $, and  $ { \phi _Y
(.) } $ are the characteristic functions of ${|h_{\text{A,D}} |^2 }$, X, and Y,
respectively, which X and Y are defined as
\begin{equation}
\begin{array}{l}
x \buildrel \Delta \over = |h_{\text{A,D}} |^2  + |h_{\text{R,D}} |^2 {\rm{   }}\\
y \buildrel \Delta \over = |h_{\text{A,D}} |^2  + z ,\\
 \end{array}
\end{equation} 
where $z = \frac{{|h_{\text{B,D}} h_{\text{R,D}} |^2 }}{{|h_{\text{B,D}} |^2  + |h_{\text{R,D}} |^2 }}.$

Due to independency between $ { |h_{\text{A,D}} |^2 } $ and $z$, we
simply have
\begin{equation}\label{eq:phixy}
\begin{array}{l}
\phi_x (S) = \phi _{|h_{\text{A,D}} |^2 } (S)\phi _{|h_{\text{R,D}} |^2 } (S) \\
\phi_y (S) = \phi _{|h_{\text{A,D}} |^2  + z} (S) = \phi _{|h_{\text{A,D}} |^2 } (S)\phi_z (S).\\
 \end{array}
\end{equation}
Then, since ${|h_{\text{A,D}} |^2 }$ and  ${|h_{\text{R,D}} |^2 }$ are exponentially distributed with parameter $\Omega_{\rightarrow D}$, their characteristic functions
are given by
\begin{equation}\label{eq:phiad} 
\phi _{|h_{\text{A,D}} |^2 } (S)= \phi _{|h_{\text{R,D}} |^2 } (S)= \frac{\Omega_{\rightarrow D}}{S +\Omega_{\rightarrow D}}.
\end{equation} 
For deriving the characteristic function of $z$ we need to know its PDF, so we propound the following lemma.

\begin{lemm}\label{lem:probability-lemma} 
Suppose that $u$ and $v$ are two independent exponentially-distributed random variables as follows
\begin{equation}
\begin{array}{l}
 u \sim \lambda \exp ( - \lambda_u u) \\
 v \sim \lambda \exp ( - \lambda_v v). \\
 \end{array}
\end{equation}
Define $z=\frac{uv}{u+v}$. Then, we have
\begin{equation}\label{eq:fzz}
f_z (z) = 2\lambda e^{ - 2\lambda z} \{ 2\lambda zK_1 (2\lambda z) +
2\lambda zK_0 (2z\lambda )\} ,z \ge 0{\rm{    }}
\end{equation}
where  ${K_0}$ and ${K_1}$  are modified Bessel functions of second kind.
\end{lemm} 

\begin{proof}
Please refer to Appendix~\ref{app:proof-of-lem}.
\end{proof} 

We use Lemma~\ref{lem:probability-lemma} with $u=|h_{\text{A,D}}|^2$ and $v=|h_{\text{B,D}}|^2$. Consider the following equality~\cite{integraltables}

\begin{equation}\label{eq:equality}
\begin{array}{l}
\int_0^\infty  {z^{\mu  - 1} e^{ - \alpha z} K_b (\beta z)dz =
\frac{{\sqrt \pi  (2\beta )^b }}{{(\alpha  + \beta )^{\mu  + b} }}}
\frac{{\Gamma (\mu  + b)\Gamma (\mu  - b)}}{{\Gamma (\mu  +
\frac{1}{2})}}F(\mu  + b,b + \frac{1}{2};\mu  +
\frac{1}{2};\frac{{\alpha  - \beta }}{{\alpha  + \beta }}){\rm{ }}\\
{Re{(\mu)}>{|Re{(b)}|},Re{(\alpha+\beta)}>0}\\
\end{array}
\end{equation} 
where ${K_v(x)} $ is the modified Bessel function of second kind.\footnote[1]{$K_v(x)$ is the solution of second-order ordinary differential equation
\[ {x^2}{\frac{{d^2}{y}}{{d}{x^2}}}+{x}\frac{{d}{y}}{{d}{x}}-{(x^2+v^2)}{y} = 0.\]}Employing equations~\eqref{eq:fzz}-\eqref{eq:equality}, the characteristic function of $Z$ can
be directly derived, as follows  
\begin{equation}\label{eq:phiz} \phi_{z}(S)=
4\Omega_{\rightarrow D}^2\frac{{\sqrt \pi  ({4\Omega_{\rightarrow D}} ) }}{{(4\Omega_{\rightarrow D}+S)}^3 }
\frac{\Gamma (3)\Gamma (1)}{\Gamma (2.5)}F(3,1.5;2.5;\frac{{S
}}{{4\Omega_{\rightarrow D}+S }}).\end{equation}

Substituting equations~\eqref{eq:phiad} and \eqref{eq:phiz} in \eqref{eq:phixy}, $\phi_{x}(S)$ and $\phi_{y}(S)$ will be derived. Then consequently by
substituting these results in equations~\eqref{eq:subeq-prob-s02}-\eqref{eq:subeq-prob-s3} and by using equation~\eqref{eq:s3-prob}, the desired bound is simply obtained.
\section{Proof of Lemma~\ref{lem:probability-lemma}}\label{app:proof-of-lem}
We write

\[
\begin{array}{ll}
 F_Z (z ) &= \Pr (\frac{{uv}}{{u + v}} \le z ) = \Pr (uv \le z (u + v)) = \int {\Pr (uv \le z (u + v)|v = \xi )} {\rm{ }}f_V (\xi )d\xi  \\
  &= \int_{\xi  = 0}^z  {\Pr (u \ge \frac{{z \xi }}{{\xi  - z }}|v = \xi )} {\rm{ }}f_V (\xi )d\xi  + \int_{\xi  = z }^\infty  {\Pr (u \le \frac{{z \xi }}{{\xi  - \alpha }}|v = \xi )} {\rm{ }}f_V (\xi )d\xi\\
 \end{array}
\]
Then, because of independency of  $u$ and $v$, we have

\[
\Pr (u \le \frac{{\alpha \xi }}{{\xi  - z }}|v = \xi ) = \Pr (u
\le \frac{{z \xi }}{{\xi  - z }}) = 1 - e^{ - \lambda _u\frac{{z \xi }}{{\xi  - \alpha }}}
\]
and
\[
\Pr (u \ge \frac{{z \xi }}{{\xi  - z }}|v = \xi ) = \Pr (u
\ge \frac{{z \xi }}{{\xi  - z}}) = 1
\]
We may then write
\[
\begin{array}{ll}
\Pr (\frac{{uv}}{{u + v}} \le z ) &= \int_{\xi  = 0}^z f_v (\xi )d\xi  + \int_{\xi  = z}^\infty  {\lambda_v e^{ -\lambda _v \xi } (1 - e^{ - \lambda_u\frac{{z \xi}}{{\xi  - z }}} )d\xi } \\&= (1 - e^{ - \lambda _v z} ) +\int_{\xi  =z }^\infty \lambda_v e^{ - \lambda _v \xi }(1 - e^{ -\lambda_u\frac{{z \xi }}{{\xi  - z }}})d\xi \\
\end{array}
\]
Then, after rearrangement, we obtain
\begin{equation}
\begin{array}{ll}
 \Pr (\frac{{uv}}{{u + v}} \le z ) &= (1 - e^{ - \lambda _v z } ) + \int_{\xi  = z }^\infty  {\lambda _v e^{ - \lambda _v \xi } } d\xi  - \int_{\xi  = z }^\infty  {\lambda _v } e^{ - (\lambda _u \frac{{z \xi }}{{\xi  - z }} + \lambda _v \xi )} d\xi  \\
  &= (1 - e^{ - \lambda _v z } ) + e^{ - \lambda _v z }  - \int_{\xi  = 0}^\infty  {\lambda _v e^{ - (\lambda _u \frac{{z (\xi  +z)}}{\xi } + \lambda_v (\xi  + z ))} } d\xi  \\
 &= 1 - \lambda_v e^{ - (\lambda _u z  + \lambda _v z )} \int_{\xi  = 0}^\infty  {e^{ - (\lambda _u (\frac{{z^2 }}{\xi }) + \lambda _v (\xi ))} } d\xi  \\
 \end{array}
\end{equation}
We use the following equality~\cite{integraltables}
\[
\int_0^\infty  {x^{\eta  - 1} e^{ - \frac{\beta }{x} - \gamma x} }
dx = 2(\frac{\beta }{\gamma })^{\frac{\eta }{2}}K_\eta  (2\sqrt
{\beta \gamma } ){\rm{ }}\\
,{[Re{(\beta)}>0, Re{(\gamma)}>0]}\\
\]
where ${K_\eta  (x) }$ is the modified Bessel function. So by
setting

\[
\eta  = 1{\rm{ }},{\rm{ }}\beta  = z^2 \lambda _u,\gamma  = \lambda _v
\]
and by ${ \lambda _v  = \lambda _u  = \lambda }$, we easily obtain

\[
\Pr (\frac{{uv}}{{u + v}} \le z ) = 1 - 2\lambda z e^{ -
2z\lambda } K_1 (2z \lambda )
\]
leading to

\[
F_Z(z) = 1 - 2\lambda ze^{-2\lambda z} K_1 (2z\lambda ),z \ge 0
\]
\section{Proof of Theorem~\ref{thm:result2}}\label{app:thmresult2}

Bit error rate (BER) in a BSC channel with error probability $p$
where users use a convolutional code at high SNR is given by~\cite{digitalcomm}

\begin{equation}\label{eq:codedpbcited}
p_b  \simeq \frac{1}{k}B_{d_{free} } 2^{d_{free} }
p^{\frac{{d_{free} }}{2}}
\end{equation}

Similar to the proof of Theorem~\ref{thm:result1}, we compute the BER conditioned on
the 4 states  $s_i$'s of the relay node. Briefly, we compute the
conditional bit error probability for case 3. By substituting
equation \eqref{eq:s3-prob} in \eqref{eq:codedpbcited} we have

\begin{equation}
P_{b|s_3 }  = \frac{1}{k}B_{d_{free} } 2^{d_{free} }\text{E}\{
\frac{1}{2}(Q(\sqrt {2\text{SNR}_{u_{max}^{1} } } ))^{\frac{{d_{free} }}{2}}  +
\frac{1}{2}(Q(\sqrt {2\text{SNR}_{{u_{max}^{2}} } } ))^{\frac{{d_{free} }}{2}} \}
\end{equation}
Using ${ Q(x) \le \frac{1}{2}e{}^{ - \frac{{x^2 }}{2}} } $ we
have

\begin{equation}
P_{b|s_3 }  \le \frac{1}{{2k}}B_{d_{free} } 2^{d_{free} }{\rm{
}}E\{ \frac{1}{{2^{\frac{{d_{free} }}{2}} }}e^{ - \text{SNR}_{u_{max}^{1}}\frac{{d_{free}
}}{2}}  + \frac{1}{{2^{\frac{{d_{free} }}{2}} }}e^{ -
\text{SNR}_{{u_{max}^{2}}}\frac{{d_{free} }}{2}} \}
\end{equation}

By the same way, the conditional BEP values for cases 1 and 2 can be
computed by substituting equations \eqref{eq:s0-prob} and \eqref{eq:s1-prob} in \eqref{eq:codedpbcited}.Then by
substituting the conditional BER in \eqref{eq:s3-prob} remaining of the proof is
like the proof of the Theorem 1.





\newpage

\end{document}